\documentclass[10pt]{article}

\usepackage{cite}

\usepackage[pdftex]{graphicx}
\graphicspath{{.}}
\DeclareGraphicsExtensions{.pdf}

\usepackage[cmex10]{amsmath}
\usepackage{amsfonts,amssymb, amsthm}
\usepackage{algorithmic}

\usepackage{array}

\usepackage{mdwmath}
\usepackage{mdwtab}

\usepackage[tight,normalsize,sf,SF]{subfigure}

\usepackage{fixltx2e}

\usepackage{url}

\newtheorem{theorem}{Theorem}
\newtheorem{property}{Property}
\newtheorem{lemma}{Lemma}
\newtheorem{definition}{Definition}

\title{The Saturated Subpaths Decomposition in $\mathbb{Z}^2$:\\a short note on generalized Tangential Cover}
\author{Fabien~Feschet\\%
		\small Universit\'e Clermont Auvergne, CNRS, SIGMA Clermont,\\
		\small Institut Pascal, F-63000 CLERMONT-FERRAND, FRANCE\protect\\
		\small E-mail: fabien.feschet@dg-medical.fr
	   }
\date{}

\begin{document}
\maketitle

\begin{abstract}
In this short note, we generalized the Tangential Cover used in Digital Geometry in order to use very general geometric predicates. We present the required notions of saturated $\alpha$-paths of a digital curve as well as \textit{conservative} predicates which indeed cover nearly all geometric digital primitives published so far. The goal of this note is to prove that under a very general situation, the size of the Tangential Cover is linear with the number of points of the input curve. The computation complexity of the Tangential Cover depends on the complexity of incremental recognition of geometric predicates. Moreover, in the discussion, we show that our approach does not rely on connectivity of points as it might be though first.\\

\end{abstract}

\section{Introduction}
The Tangential Cover was first introduced as an algorithm in \cite{FT99} and fully recognized as a fundamental structure in subsequent publications \cite{FF05,Jaco05}. Basically the Tangential Cover was the set of maximal Digital Straight Segments (DSS) which can be built on an input digital curve. Since its complexity is linear, computing the Tangential Cover is practically feasible and permits to tackle various problems on digital curves such as tangent computations or polygonalizations. The tangents obtained using the Tangential Cover does indeed converge towards their real counterpart for convex function. If the Tangential Cover was introduced using only DSS, it has also been used on extensions of DSS, such as thick segments \cite{Debxx}. Moreover, thanks to the local notion of meaningful thickness \cite{kerxx}, the Tangential Cover has been made adaptive allowing to perform analysis of noisy digital curves using automatic multi-scale selection \cite{Debyy}. It has been also included in the DGtal library \cite{DGtal}. As it can be seen in the comments of the description of the source code of the DGtal library, extensions of the Tangential Cover is not straightforward in any case. In this paper, we explain what is a good notion of geometric predicates and proves that the Tangential Cover can be built with general predicates called conservative predicates still resulting in a linear output size complexity. We moreover make a careful clarification of the input of the Tangential Cover recalling that the Tangential Cover works for non simple curves as well as non connected sets of points. We hope that this short note will clarify the usage of the Tangential Cover. 

\section{Background notations}
We consider, in this paper, points in the digital plane $\mathbb{Z}^2$
and some adjacency $\alpha$ on those points. The usual adjacencies are
the classical 4- and 8-adjacencies \cite{KleRos04}. Let us recall that
if we denote by $\mathcal{N}()$ a norm in $\mathbb{R}^2$, they are
defined by the fact that two points $p$ and $q$ are adjacent if and
only if $\mathcal{N}(q-p) \leq 1$. The norm $\mathcal{N}()$ is the
$l_1$ norm - corresponding to the sum of the absolute differences of
the coordinates of the points - for the 4-adjacency and it is the
infinite norm $l_{\infty}$ - corresponding to the max of the absolute
differences of the coordinates of the points - for the 8-adjacency. In
the sequel, the adjacency is fixed. Of course, our work readily
extends to the case of other cellular decompositions of the real plane
$\mathbb{R}^2$ such as the triangular or the hexagonal decomposition
but we restrict our presentation to the classical rectangular grid.

A digital $\alpha$-path, shortly a digital path, in the plane is a
list of points $\left(p_0,\ldots,p_n\right)$ of $\mathbb{Z}^2$ such
that $p_i$ and $p_{i+1}$ are $\alpha$-adjacent and $p_i\neq
p_{i+1}$. It must be noted that a path is indeed an ordered set of
points of $\mathbb{Z}^2$ such that two consecutive points are
$\alpha$-adjacent. Hence, for a given set of points, they may be
several paths describing it depending on the chosen ordering. A path
is closed if and only if $p_0$ and $p_n$ are $\alpha$-adjacent. When a
path is closed, indices of the points in the path are intended modulo
$n+1$, the length of the path. For a path $P$, a subpath is a sublist
$(p_i,\ldots,p_j)$ of points of $P$ which is an $\alpha$-path. Care
must be taken that $\alpha$-connected subsets of $P$ might not be
subpaths (see Fig. \ref{fig:subpath}) since we impose on each subpath
to keep all points between $p_i$ and $p_j$ with respect to the
ordering of the initial list describing the reference path.

\begin{figure}
  \centering
  \scalebox{0.75}{%
    \includegraphics{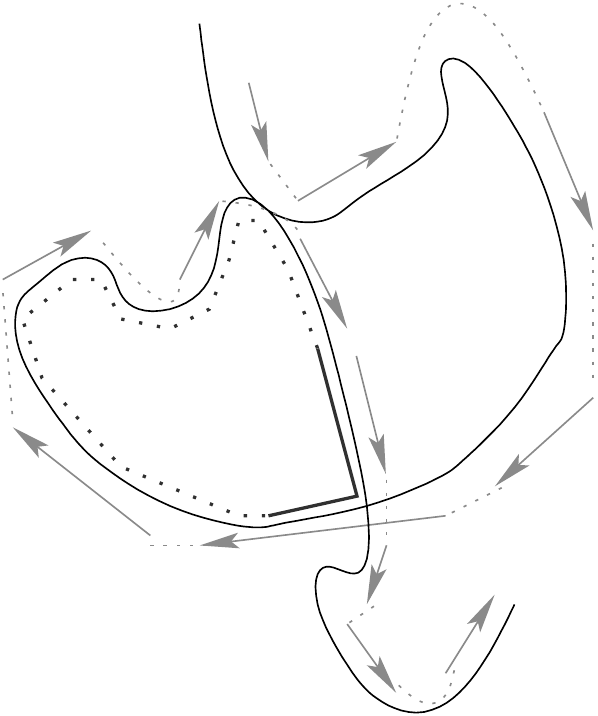}%
  }
  \caption{An $\alpha$-path with the order given by the arrows. The
    $\alpha$-connected subpart depicted in bold line is not a subpath
    since the dotted part (in bold) is missing in the set.}
  \label{fig:subpath}
\end{figure}

For an $\alpha$-path $P$, either closed or not, we denote by
$\mathcal{SP}(P)$ the set of all subpaths of $P$. This set contains
$O(n^2)$ elements by construction which is far less than the number
($O(2^n)$) of subsets of $P$.

A closed digital curve is a digital path where each point $p_i$ has
exactly two $\alpha$-adjacent points in the path. An open digital
curve is a digital path with the same property except for the points
$p_0$ and $p_n$ who have only one $\alpha$-adjacent point in the
path. These notions avoid the possibility of self-intersections and
are, thus, too restrictive in practice. Hence, otherwise precisely
mentioned, all digital sets manipulated in the paper will be digital
paths.

\begin{figure}
  \centering
      \scalebox{0.45}{%
        \includegraphics{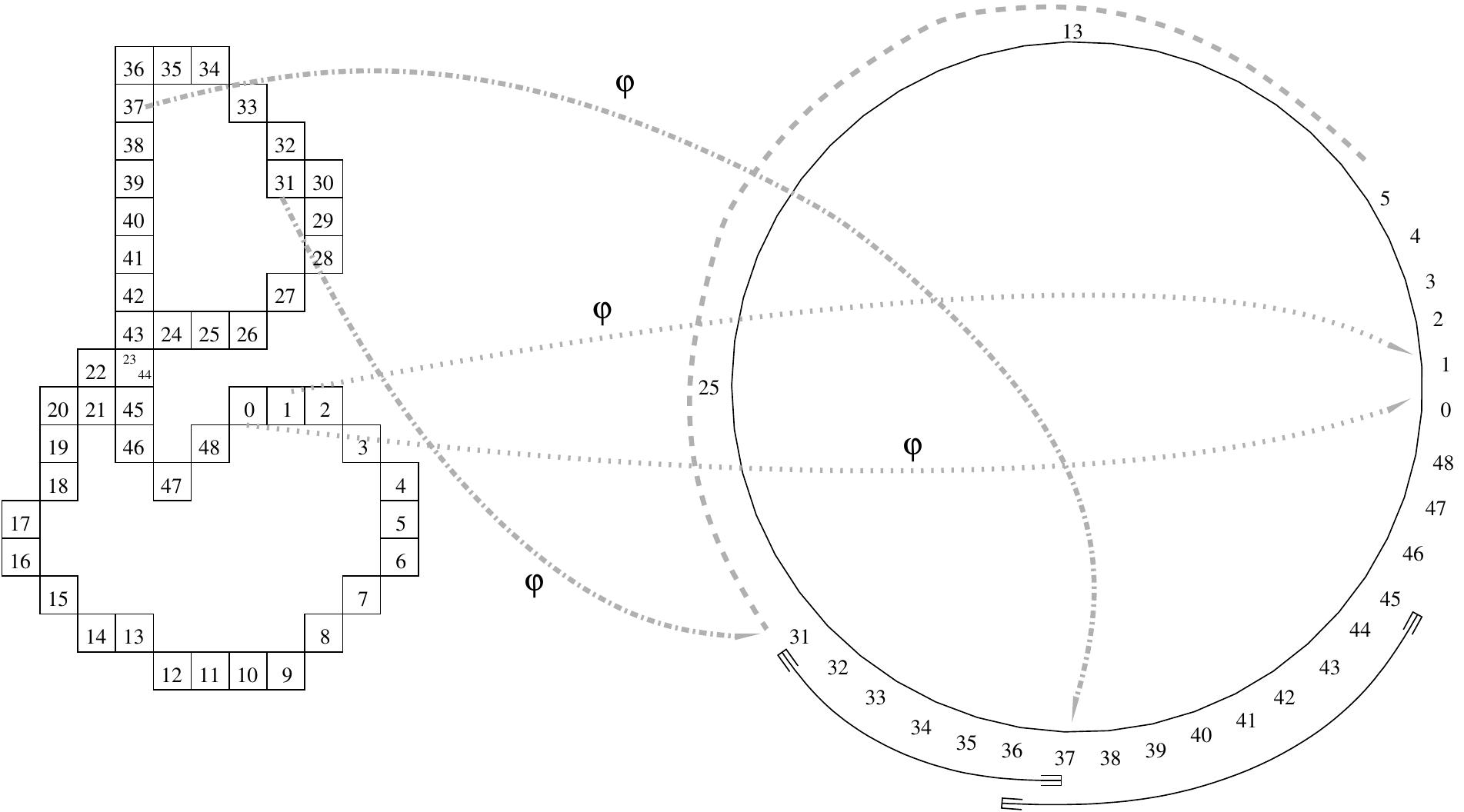}%
      }
  \caption{(left) An $\alpha$-path with the order written inside the
    pixels (right) The mapping of the points and the overlapping
    (31,37) and (36,45) subpaths.}%
  \label{fig:defs}
\end{figure}

Our main goal is to represent digital path as a graph structure. Since we only
consider subpaths, they are always connected subsets by definition. So, we
have a natural mapping $\varphi$ between any $\alpha$-path and a circular arc
graph as follows. To do this, we fix an $\alpha$-path $P$ given by its list of
points $\left(p_0,\ldots,p_n\right)$. Every point in the path is sent via the
mapping $\varphi$ to a point on the unit circle via the formula: $\varphi(p_k)
= \exp(2i\pi k/(n+1))$ where $i^2 = -1$. This mapping can then be extended to
map paths onto arcs of the circle. Indeed, any subpath of $P$ is given by a
list of consecutive points $\left(p_{i},\ldots,p_{j}\right)$ as shown on
Fig.~\ref{fig:defs}. The $\varphi$-image of the list is an arc given by
$\left(\varphi(p_{i}),\varphi(p_{j})\right)$. The $\varphi$ mapping permits us
to map any $\alpha$-path onto a circular arc graph. This graph is an interval
graph only when the path is open. It should be noted that $\varphi$ is an
inclusion preserving mapping that is, if $S$ is a subpath of $P$ then
$\varphi(S)$ is included in $\varphi(P)$.

As it can be seen in the definition of an $\alpha$-path, the first
problem to solve in our approach is to build an $\alpha$-path from an
image. Then, the $\varphi$ mapping can be used to construct a graph
from the path. Hence, in the next section, we explain in detail how to
build our reference $\alpha$-path from an image.

\section{From images to path}
The notion of $\alpha$-path is intrinsically related to the notion of ordering
of points. Indeed, we can define the \textbf{canonical} extension of an
$\alpha$-path $P$ corresponding to the list of points
$\left(p_0,\ldots,p_n\right)$ by the piecewise linear function
$P_{\text{can}}$ corresponding to the real polygon passing through the points
from $p_0$ to $p_n$ and parametrized on $[0,1]$. As each consecutive points
are different, this clearly defines a function for which $P_{\text{can}}(0) =
p_0$ and $P_{\text{can}}(1) = p_n$. Hence, this function is continuous,
point-wise and preserves the order between the real segment $[0,1]$ and the
path $P$. As such, this function is a Jordan line \cite{Frechet28}~(pp. 149-150) and \cite{Alexandrov89} (chapter 1). So, when starting from an image, the only problem to solve is
to build an ordering of the points in the image as this completely defines a
Jordan line. W.l.o.g., we can suppose that there is only one connected
component with respect to the $\alpha$-adjacency in the image because if it is
not, all connected components can be treated separately.

We base our approach on the curve tracing algorithm with self-intersections
considered by Hajdu and Pitas \cite{Haj08}. The key point of their approach
is the notion of junction \cite{KleRos04,Kle06}. We now briefly present their
approach as well as its straightforward rewriting for our goal of constructing
$\alpha$-paths. It should be noted that as is, the method proposed by Hajdu and
Pitas \cite{Haj08} is in fact somewhat imprecise as we will see later in this
paper. This is mainly due to the fact that they do not consider complex
junctions. We have corrected this little drawback in our version of their
method and consequently still refer to it as Hajdu and Pitas' method.

The classification of simple curve is based on the notion of
\textbf{branching} index. For a point $p\in\mathbb{Z}^2$, we consider $V(p)$
as the set of points $q \in\mathbb{Z}^2$ such that $p$ and $q$ are
$\alpha$-adjacent. $V(p)$ is the $\alpha$-neighbourhood of $p$. The branching
index of $p$ is then the number of points of $q$ which also are image
point. The classification of points in an image depending on their branching
index \cite{KleRos04} is the following:
\begin{itemize}
\item {$p$ is an \textit{end} point if its branching index is 1,}
\item {$p$ is a \textit{regular} point if its branching index is 2,}
\item {$p$ is a \textit{branching} point if its branching index is $\geq 3$.}
\end{itemize}
An $\alpha$-junction is an $\alpha$-connected subset of the image such that
each point in the junction is a branching point and such that it is maximal
for the inclusion order\footnote{The maximality with respect to inclusion is
  essential in the definition of a junction but is missing in \cite{Haj08} and
  in \cite{Kle06} even if it is evident from their papers that they both do
  consider maximal junctions.}. Hence, if a junction $J_1$ contains a junction
$J_2$ then $J_1 = J_2$. By extension, the branching index of a junction is the
number of regular points which are $\alpha$-adjacent to a point in the
junction. If we consider a binary image $I$
containing only one $\alpha$-connected components and if we suppress all
junctions in $I$, we get a simplified image $\tilde{I}$. This image might have
several $\alpha$-connected components but each connected component is a simple
open curve. This process gives a decomposition of the image $I$ into simple
open curves and junctions. This naturally leads to a graph $G=(V,E)$ as a
representation of $I$ as follows:
\begin{itemize}
\item {$V$ is composed of all end points and all junctions,}
\item {$E$ is the set of undirected edges $\{u,v\}$ with $u$ and $v$ two
    vertices such that $u$ and $v$ are $\alpha$-adjacent in $I$.}
\end{itemize}
The key point in Hajdu and Pitas' algorithm is to remark that an $\alpha$-path in
the image $I$ is thus a path in the graph $G$ where each edge is used only
once. Of course, this path might cross a junction more than once. So, the
construction of an $\alpha$-path is just the decomposition of the graph $G$
into Eulerian subgraphs for which it has been proved long time ago that they
always have an Eulerian tour, that is a path taking every edge exactly
once. They also proposed a way of creating only one $\alpha$-path for the
whole connected component in $I$ if repetition is allowed, that is some parts
are used back and forth resulting in the multiple use of non branching points
in the image (see part VI in \cite{Haj08}). This last problem was shown to be
the Chinese Postman Problem, a well known problem in graph theory
\cite{Schrij03} (chap. 29 vol A).

Originally Hajdu and Pitas' method is based on straight line segments
and minimal length paths between two open curves sharing a
junction. This is natural since their goals is also to apply a polygonal
approximation algorithm on the image to efficiently compress
it. However, two points are missing in their algorithm. First, they do
not fully explain how they consider self-loop in the graph
$G$. Indeed, self-loops cannot
be distinguish in the graph $G$, hence they do not have any guarantee
that their algorithm for constructing Eulerian tours will be able to
take them in the right geometric order. This problem can have
practical drawbacks but it is not a fundamental problem, to our point
of view. The second problem with their method is that it is
incorrect in full generality. Indeed, when the Eulerian tour as been
constructed, they built an 8-connected path between pairs of open
curves. However,
the junction is not the union of the paths traversing it. Hence, only
a subpart of the junction is coded by shortest paths between in and
out points of a junction. It should be noticed that the JBEAM
algorithm \cite{Huo05}, cited in \cite{Haj08}, does not have this
problem. However, it is not suitable for our purpose since it does not
consider paths but sets of points.

\section{$\alpha$-paths decomposition}
The first step in our method consists in a decomposition of an
$\alpha$-path $P$ into meaningful subparts, such as geometric
primitives but not only. We must represent them in the most possible
generic manner. For instance, several geometric primitives have been
studied in the literature such as Digital Straight Segments (DSS),
arcs of circles, arcs of parabolas and so on. More primitives can also
be constructed and we shall derive a representation sufficiently
generic for all such different basic elements. It appears that the
most generic model is the model of predicates. A predicate, denoted by
$\mathcal{P}$, is a Boolean value function defined on
$\mathcal{SP}(P)$. We write $\mathcal{P}(X)$ if the predicate
$\mathcal{P}$ has value \textit{true} for the path $X$.

In the case of DSS, it has been proved that maximality is a basic
property of digital segments on shapes \cite{FF05,Jaco05}. We believe that if
a subpath is meaningful, but is contained into another subpath also
meaningful, we should retain only the largest one. Hence, we keep
this principle of maximality in our construction. In the context of
predicates, this notion was introduced by Mazurkiewicz \cite{Mazur20}.

\medskip

\begin{definition}
  A set $X$ is \textbf{saturated} with respect to a predicate
  $\mathcal{P}$ if
  \begin{enumerate}
  \item{$\mathcal{P}(X)$ ($X$ has the property corresponding to the
      predicate $\mathcal{P}$),}
  \item{every set $Y\supseteq X$ with $\mathcal{P}(Y)$ is such that
      $Y=X$.}
  \end{enumerate}
\end{definition}

\medskip

In other word, a saturated set is a maximal set with respect to the
inclusion order and for which the predicate is true. Clearly, when we
consider \textit{saturated path}, we also take into account the
connectivity of the points since any path is connected. As a
consequence, we can bound the number of saturated subpath for any
predicate.

\medskip

\begin{lemma}
  Let $P$ be a path with $n$ points. Then, for any predicate
  $\mathcal{P}$, there are at most $O(n)$ saturated subpath in
  $\mathcal{SP}(P)$.
\end{lemma}
\begin{proof}
  We consider $P$ as oriented positively with respect to the
  increasing order of the indices of its points. Let us consider a
  subpath $X$ of $P$ being composed of the sublist $(p_i, \ldots,
  p_{i+k})$ of points. We define its middle point as $p_m=p_{i+\lfloor
    k/2 \rfloor}$. Then $X$ is reconstructed by adding points around
  $p_m$ starting from the positive orientation. That is: $p_m$,
  $p_{m+1}$, $p_{m-1}$, ... So to any subpath $X$, we can associate
  its middle point $p_m$. This define a function $m(X)=p_m$ from
  $\mathcal{SP}(P)$ to $P$. This function is obviously onto but
  not injective as symmetric subpath around $p_m$ all have the same
  middle point. But, given $p_m$ the subpaths of $P$ having $p_m$ as
  middle points are totally ordered by inclusion. This is due to the
  fact that a subpath has no hole in the indices of its points.
  Hence, if we restrict the domain of $m()$ onto the set of saturated
  subpaths then it becomes injective but loose its
  surjectivity. Indeed, for an arbitrary predicate it is possible that
  no sets $M$ containing $p_m$ might be such that
  $\mathcal{P}(M)$. For instance the predicate \textit{``to be a set
    containing $p_0$''} clearly has only one saturated subpath which
  contains $p_0$. So since $m()$ is injective, a point in $P$ can not
  be used twice as a middle of a saturated subpath and since the
  function is not onto, not all points might be used. Hence, the
  number of points covered by the $m()$ function restricted on
  saturated subpath is clearly bounded by $n$, hence the number of
  saturated subpaths is also bounded by $n$.
\end{proof}
This property is false when applied to sets. For instance, the predicate
\textit{``to have 3 points''} has $O(n^3)$ saturated sets but only $O(n)$
saturated subpaths on $P$. It is important to notice that even if there are
$O(n^2)$ subpaths in $P$, only $O(n)$ at most can be saturated. This implies
that using saturated sets, algorithms with small, indeed linear if possible,
time complexity are more easily reachable.

It is straightforward to see that for saturated subpaths, the inclusion is
forbidden such that the corresponding graph is necessarily proper.

\begin{property}
  The circular-arc graph constructed via the $\varphi$ mapping of the
  saturated subpaths of any predicate $\mathcal{P}$, is a proper circular-arc
  graph.
\end{property}

Let us now describe a simple algorithm which has been introduced in
the context of maximal DSS \cite{FT99}. Our goal is to describe how to
generalize this algorithm in the case of predicates. W.l.o.g., we
suppose the path $P$ to be closed. We start with an arbitrary point,
for instance $p_0$. Let us denote by $S$, the current tested set. The
algorithm can be summarized by the following steps.
\begin{enumerate}
\item {\textit{Increase:} add points to $S$ one by one, always
    starting in the positive orientation. Hence, if $S$ is a subpath
    of the form $\left(p_{m-k},\ldots,p_{m+k}\right)$ then first add
    $p_{m+k+1}$ before $p_{m-k-1}$. If $S$ is of the form
    $\left(p_{m-k+1},\ldots,p_{m+k}\right)$, then add $p_{m-k}$.}
\item {\textit{Check:} if $S$ is not a DSS then remove the last added
    point to $S$, else go back to the \textit{Increase} step.}
\item {\textit{Maximality:} Once $S$ is symmetric and maximal then go
    back to the \textit{Increase} step but only for one
    orientation. In case of ambiguity, the positive orientation must
    be chosen. When the \textit{Increase} step stops again, move to
    the \textit{Restart} step.}
\item {\textit{Restart:} $S$ is a maximal DSS. Suppose that $p_{m+r}$
    is its last point. Then let $S^+ = S \cup \{ p_{m+r+1}\}$. Remove
    points at the negative side of $S^+$ until it becomes a DSS. Then
    go to the \textit{Increase} step to compute another maximal DSS.}
\end{enumerate}

\medskip

This algorithm is proven to build the set of all maximal DSS of the
path $\mathcal{P}$ in linear time \cite{FT99}. It is important to
notice why the preceding algorithm is correct. Two facts are
implicitly used,
\begin{enumerate}
\item {If $S$ is not a DSS then any set $T$ containing $S$ is not a DSS ;}
\item {Two consecutive maximal DSS shares at least one point.}
\end{enumerate}

\medskip

The first property is used at the step \textit{Check} of the algorithm to
break the symmetry of $S$ knowing that $S$ cannot be extended any more on both
sides and it is used at the step \textit{Maximality} because when the one side
extension of $S$ has been built, it is declared to be a maximal DSS. The
second property is used in the step \textit{Restart} to ensure the fact that
the decreasing of $S$ always lead to a DSS and so the algorithm can start
again.

Careful must be taken in the extension of this algorithm to
predicates. Indeed, if we consider any predicate, the question of
finding a subpath $X$ such that $\mathcal{P}(X)$ probably could
necessitate to test all subpaths leading to a quadratic complexity. Of
course, this is practical for small $n$ but our goal is to preserve
the linear complexity of the algorithm during its extension, while
keeping a sufficiently rich class of predicates containing at least
classical geometric primitives.


We will restrict our attention to a special class of predicates in order to
have interesting theoretical properties on the saturated subpaths. To do so,
we only consider the notion of \textbf{conservative} predicate.

\medskip

\begin{definition}
  A predicate $\mathcal{P}$ is conservative if and only if
  $\mathcal{P}(X)$ for a path $X$ then for all $Y\in\mathcal{SP}(X)$,
  $\mathcal{P}(Y)$.
\end{definition}

\medskip

In other words, when an conservative predicate is true on a path, it
is true on all its subpaths. Of course, this also implies ab absurdo
that if $X$ is such that $\neg\mathcal{P}(X)$ then for any set $Y$
containing $X$, we have $\neg\mathcal{P}(Y)$. One important
consequence is that if we know a set $X$ for which the predicate is
true or false then we can decide if there is a saturated set
containing $X$. Hence, the simplest sets to test are the
singletons. Indeed for any point $p$ in $P$, if $\mathcal{P}(\{p\})$
then there exists a saturated set containing $p$ ; else, it is clear
that no saturated set contains the point $p$. This simple property is
sufficient to modify the classical algorithm to obtain an algorithm
valid for any conservative predicates.

It seems important to insist on the fact that the conservative
property is a very natural property for geometric predicates. Indeed,
if a path is a DSS, or an arc of a circle, or arc of a parabola then,
it is natural to have that all subpaths are all DSS, arcs of circles
or arcs of parabolas. The same is true for the negative case: if a set
is not a DSS then clearly no superset can be a DSS !

\medskip

The steps of the new algorithm are as follows.

\begin{enumerate}
\item {\textit{Init:} Let $p$ be the current point. While
    $\neg\mathcal{P}(\{p\})$, go to the next point. Stop either when
    all points in $P$ have been tested or if a point $p$ such that
    $\mathcal{P}(\{p\})$ has been found. In this case go to the
    \textit{Increase} step.}
\item {\textit{Increase:} this step is unmodified.}
\item {\textit{Check:} this step is unmodified (simply use predicate
    $\mathcal{P}$ and not a DSS test).}
\item {\textit{Maximality:} This step is unmodified.}
\item {\textit{Restart:} $S$ is a maximal subpath verifying
    $\mathcal{P}$. Suppose that $p_{m+r}$ is its last point. Then let
    $S^+ = S \cup \{ p_{m+r+1}\}$. If $\neg\mathcal{P}(\{p_{m+r}\})$
    then set $p=p_{m+r+1}$ and go to the \textit{Init} step. Else
    remove points at the negative side of $S^+$ until it verifies
    $\mathcal{P}$. Then go to the \textit{Increase} step to compute
    another saturated subpath.}
\end{enumerate}

\medskip

It is clear that this algorithm, referred as \texttt{SSD} (Saturated Subpaths
Decomposition), can be improved by a test on all singletons. Indeed, each time
a point $p$ is used, if $\neg\mathcal{P}(\{p\})$ then no saturated set
containing $p$ can exist such that we can directly move to the next point in
$P$. This optimization does not change however the complexity of the
algorithm. For the same reason than in \cite{FT99}, this new version of the
classical algorithm still has a linear time complexity in terms of the number
of time the predicate is checked. Hence, for constant time predicate, this
extension still runs in linear time.

\begin{theorem}
  The algorithm \emph{\texttt{SSD}} is correct, that is, it computes all
  saturated subpaths of a conservative predicate $\mathcal{P}$.
\end{theorem}
\begin{proof}
  Let us suppose w.l.o.g. that $\mathcal{P}(\{p_0\})$. There are two
  cases: either there exists a saturated subpath $S_0$ such that
  $m(S_0)=p_0$, or there is a saturated subpath $S'_0$ which
  contains $p_0$ but not as a middle. In the second case, several such
  subpaths might exist and $S'_0$ is chosen to be the one with the
  rightmost middle. In the first case, the \textit{Increase} and
  \textit{Check} steps are built to find $S_0$ so the algorithm is
  correct if and only if it correctly moves to the next saturated
  subpath. 

  In the second case, the \textit{Increase-Check} steps build a set
  $S$ whose middle is $p_0$ such that $\mathcal{P}(S)$. Then the
  \textit{Maximality} step is used. Suppose that $S$ has the form
  $(p_i,\ldots,p_j)$. There are two cases : either
  $\mathcal{P}(\{p_i,\ldots,p_{j+1}\})$ or
  $\neg\mathcal{P}(\{p_i,\ldots,p_{j+1}\})$. In the first case, the
  \textit{Maximality} step implies that the \textit{Increase} step is
  pursued on the positive side and in the second case, on the negative
  side. On both sides, a saturated set is built with a rightmost
  middle, that is $S'_0$. Hence, once again, the algorithm is correct
  if and only if it correctly moves to the next saturated subpath.
  
  So the key point is the \textit{Restart} step. Let
  $S=(p_i,\ldots,p_j)$ be the last computed saturated subpath. Either
  $\mathcal{P}(\{p_{j+1}\})$ or $\neg\mathcal{P}(\{p_{j+1}\})$. In the
  first case, there exists a saturated set containing $p_{j+1}$. Since
  this point was not in $S$, it is clear that its middle is at the
  right of $m(S)$. Hence the \textit{Restart} step stops necessarily
  at the left of this middle point and then the
  \textit{Increase-Check} steps are used to find the correct middle
  point. 
  
  In the second case, no saturated subpaths contain $p_{j+1}$ thus,
  the \textit{Restart} step move to $p_{j+2}$ and the \textit{Init}
  step is used to find a new point belonging to a saturated
  subpath. So, no saturated subpath can be missed and thus the
  algorithm is correct provided than it stops when
  all points have been tested.
\end{proof}

\section{Discussion}
In the proof of Theorem 1, we mainly used the middle point of a subpath. This is the common way to uniquely identified a saturated subpath. However, it is possible, in practice, to use only forward incremental recognition. Indeed, this leads to the same Tangential Cover with the only careful check of the first computed primitive. Starting at any point of the curve does not imply that this point is a beginning point of a saturated subpath. Hence, to complete the construction of the Tangential Cover, this first computed primitives should be tested for maximality and suppressed if it is not maximal. This means that the algorithm must be pursued until the end of the first recognized primitive is reached. Obviously the complexity remains the same. 

We must also remark that our proof does not rely on the adjacency $\alpha$. Indeed, we only use the indices of the points. Hence, an adjacency can be based solely on a numbering of the points. Hence, the points does not need to be connected in the usual sense, they are connected by indices only. Of course, we should add that a good geometric predicate must be true on any singletons. Giving a numbering is a more profound property than connectivity in fact. The only case to forbid is the repetition of points which consecutive indices. This might be a problem for computing the saturation of subpaths. Indeed, at first saturation of paths might be viewed as a set maximality for inclusion, but it is more general to have a maximality with respect to connected subsets of the set of indices. In that case, maximality does not imply that the points in paths are different between saturated sets but this avoid useless complexity in the decomposition.

So to conclude, the original set of points can described very general curves, allowing non simplicity and non connectivity. This is not a problem since the canonical extension $P_\text{can}$ is properly defined on both cases.

\section{Conclusion}
In this short note, we clarified the generalization of the Tangential Cover to arbitrary geometric predicates using the usual notion of conservatism. This work was done in 2009 but never published. When reading the source code of the DGtal library \cite{DGtal}, I have discovered that concepts of segments (\texttt{CForwardSegmentComputer}) where supposed to be true in specific subsets of a segment. Hence, I have decided to write this note to prove that the size complexity of this generalized Tangential Cover is still linear.


\begin{thebibliography}{1}
\bibitem[1]{KleRos04} R. Klette and A. Rosenfeld, \emph{Digital Geometry -
    Geometric Methods for Picture Analysis}, Morgan Kaufmann, San Francisco,
  2004.
\bibitem[2]{Mazur20} S. Mazurkiewicz, \emph{Sur les lignes de Jordan},
  Fun. Math., vol 1, 1920, pp. 166-209.
\bibitem[3]{FT99} F. Feschet and L. Tougne, \emph{Optimal Time
    Computation of the Tangent of a Discrete Curve: Application to the
    Curvature}, in Proc. Intl. Conf. DGCI, LNCS vol 1568, pp.~31-40,
  Springer-Verlag, 1999.
\bibitem[4]{FF05} F. Feschet, \emph{Canonical representations of
    discrete curves}, Pattern Anal. Appl., vol. 8, no. 1-2, pp.~84-94,
  2005.
\bibitem[5]{Haj08} A. Hajdu and I. Pitas, ``Piecewise Linear Digital Curve
  Representation and Compression Using Graph Theory and a Line Segment
  Alphabet'', IEEE Trans. on Image Processing vol. 17, no. 2, pp.~126-133,
  2008.
\bibitem[6]{Kle06} G. Klette, ``Branch Voxels and Junctions in 3D skeletons.'' in
  Proc. Intl. Conf. 10th IWCIA 2006, Berlin, Germany, LNCS vol. 4040, Springer,
  Berlin, pp. 34–-44, 2006
\bibitem[7]{Schrij03} A. Schrijver, \emph{Combinatorial Optimization, Polyhedra
    and Efficiency}, Springer-Verlag, Berlin, 2003.
\bibitem[8]{Frechet28} M. Frechet, \emph{Les Espaces Abstraits},
  Gauthier-Villars, 1928 (reprint by J. Gabay Edition 1989 ISBN:
  2-87647-056-X).
\bibitem[9]{Huo05} X. Huo, and J. Chen, ``JBEAM: multiscale curve coding via
  beamlets.'', IEEE Trans. on Image Processing, vol. 14, no. 11,
  pp. 1665–-1677, 2005.
\bibitem[10]{Alexandrov89} A.D. Alexandrov, and Yu G. Reshetnyak, \emph{General theory of irregular curve},
Kluwer Academic Publishers, 1989.
\bibitem[11]{Jaco05} J.{-}O. Lachaud, and A. Vialard, and F de Vieilleville, \emph{Analysis and Comparative Evaluation of Discrete Tangent Estimators}, Discrete Geometry for Computer Imagery, 12th International Conference, pp. 240--251, 2005.
\bibitem[12]{DGtal} {D}{G}tal: Digital Geometry tools and algorithms library,  \url{https://dgtal.org/}, version 0.9.4.1., June 2018.
\bibitem[13]{Debyy} P. Ngo, and I. Debled-Rennesson, and B. Kerautret, and H. Nasser, \emph{Analysis of Noisy Digital Contours with Adaptive Tangential Cover}, Mathematical Imaging and Vision 59(1): 123-135 (2017).
\bibitem[14]{Debxx} I. Debled-Rennesson, and F. Feschet, and J. Rouyer-Degli, \emph{Optimal blurred segments decomposition of noisy shapes in linear time}, Computers \& Graphics 30(1): 30-36 (2006).
\bibitem[15]{kerxx} B. Kerautret, and J-O. Lachaud, and M. Said, \emph{Meaningful Thickness Detection on Polygonal Curve}, ICPRAM (2): 372-379 (2012).

\end{thebibliography}
\end{document}